\pdfoutput=1
\RequirePackage[T1]{fontenc} %
\documentclass[letterpaper,10pt,conference]{ieeeconf}
\IEEEoverridecommandlockouts
\usepackage{amsmath,amssymb,amsthm,mathtools}

\usepackage{graphicx}
\usepackage{booktabs}
\usepackage{needspace} %
\graphicspath{{figures/}}
\usepackage[draft]{hyperref} %

\newtheorem{theorem}{Theorem}
\newtheorem{lemma}{Lemma}
\newtheorem{proposition}{Proposition}
\newtheorem{corollary}{Corollary}
\theoremstyle{definition}
\newtheorem{definition}{Definition}
\newtheorem{assumption}{Assumption}
\newtheorem{remark}{Remark}

\newcommand{\R}{\mathbb{R}}
\newcommand{\dd}{\,\mathrm{d}}
\newcommand{\md}[1]{\mathcal{M}_{#1}}
\newcommand{\Jint}{J_{\mathrm{int}}}
\newcommand{\Jesc}{J_{\mathrm{esc}}}
\newcommand{\vmax}{v_{\max}}
\newcommand{\phimax}{\varphi_{\max}}
\newcommand{\eps}{\varepsilon}
\newcommand{\cFanDrive}{1.5}
\newcommand{\cFanCapacityTime}{0.1}
\newcommand{\cFanAltFailureTime}{0.05}

\newcommand{\cJint}{0.5039}
\newcommand{\cJconstruct}{0.511831}
\newcommand{\cCertDrives}{0.28415,\ 2,\ 0.28215}
\newcommand{\cCertDurations}{1.6962,\ 0.49971,\ 2.5298}
\newcommand{\cCertTw}{4.72571}

\newcommand{\cCertGapPercent}{1.57}

\newcommand{\cPersistRadius}{0.051}
\newcommand{\cPersistBound}{0.0473}

\newcommand{\cHorizon}{2}
\newcommand{\cDzero}{9.03\times 10^{-5}}

\newcommand{\cThetaSum}{0.1}

\newcommand{\cNsamp}{100}

\newcommand{\cEpsH}{10^{-6}}

\title{Restorative Control and the Limits\\ of Safe Model Discrimination}

\author{Micheal Myers Jr.%
\thanks{%
This material is based upon work supported by the U.S. National Science Foundation CISE Graduate Fellowships under Grant No. 2313998.
Any opinions, findings and conclusions or recommendations expressed in this material are those of the author and do not necessarily reflect the views of the U.S. National Science Foundation.
The author is with the School of Computational Science and Engineering, Georgia Institute of Technology, Atlanta, GA. Email \texttt{mmyers75@gatech.edu}.
Codex and Claude \cite{openai-codex,anthropic-claude-code} assisted with manuscript revisions, proof checking, and numerical code. The author is responsible for all results.
Code and numerical results are available at \url{github.com/michealrm/safe-restoration}.}}

\AtBeginDocument{%
\typeout{PAPERLAYOUT top=\the\dimexpr1in+\topmargin+\headheight+\headsep\relax}%
\typeout{PAPERLAYOUT height=\the\textheight}%
\typeout{PAPERLAYOUT left=\the\dimexpr1in+\oddsidemargin\relax}%
\typeout{PAPERLAYOUT column=\the\columnwidth}%
\typeout{PAPERLAYOUT gap=\the\columnsep}%
\typeout{PAPERLAYOUT baseline=\the\baselineskip}}

\begin{document}
\bstctlcite{BSTcontrol}
\nocite{clark-destiny,haehnel-cml,kantarjian-enestnd}
\maketitle

\begin{abstract}
Restorative control seeks lasting recovery through an intervention that ends in finite time. We formulate this objective as a reach-avoid problem in which the input must bring the system safely into a desired basin of attraction before withdrawal. Two candidates share bistable baseline dynamics but accumulate irreversible damage through different allocations of drive-generated exposure. We derive a lower bound on the exposure required for basin transfer and a sufficient construction for safe restoration. When the lower bound exceeds the summed tolerances, every common open-loop input that restores either model causes failure under at least one. With sufficient tolerance, proportional allocation restores both without identification. For feedback from output samples with Gaussian noise, we bound safe-restoration probability in terms of failure risk under the alternative model. A numerical example illustrates the certificates under two sensing regimes.

\end{abstract}

\section{Introduction}\label{sec:intro}

Restorative control aims to recover a desired operating regime that persists after a finite intervention ends. One example is treatment-free remission in chronic myeloid leukemia (CML), in which disease control persists after tyrosine kinase inhibitor treatment ends. In the DESTINY clinical study, patients reduced their dose before stopping treatment \cite{clark-destiny}. A CML model by H\"ahnel et al. represents successful withdrawal as transfer into a remission basin. Simulations based on retrospective fits show how responses to dose reduction can help predict withdrawal outcomes \cite{haehnel-cml}. Treatment before withdrawal also carries a risk of harm. In the ENESTnd randomized CML trial, cardiovascular events were reported over ten years in 16.5\% and 23.5\% of patients in the two nilotinib dose groups, compared with 3.6\% with imatinib \cite{kantarjian-enestnd}. Together, these studies motivate studying lasting recovery, treatment harm and learning.

We study the underlying control question through a stylized model of irreversible damage and uncertain routing. The controller applies a drive $v$, such as dose intensity, and splits it between two routes, such as two inhibitors. Each of two candidate models directs the resulting damage to a different vulnerable component, such as an organ, with finite tolerance. The construction asks when uncertainty about the safe route prevents lasting recovery, even though each known candidate can be restored.

This is a reach-avoid task that brings the system into a recovery basin while keeping intervention-induced damage below its tolerance. For withdrawal time $T_w$ and desired operating set $H$, the objective is
\begin{equation}\label{eq:restorative-spec}
v(t)=0\quad \forall t>T_w,
\qquad y(T_w)\in\mathcal A(H).
\end{equation}
Here, $\mathcal A(H)$ is the basin of attraction of $H$ for the applicable candidate and damage remaining at withdrawal. This damage fixes the residual drag and shifted separator.

The candidates share bistable baseline output dynamics. Drive produces exposure, allocated as damage $w_1$ through fraction $p$ under $\md1$, or as $w_2$ through $1-p$ under $\md2$. Failure occurs at tolerance $\Theta_i$. Figure~\ref{fig:mechanism} summarizes these relations. The contributions of this paper are:

\begin{figure*}[t]
\centering
\includegraphics[width=\textwidth]{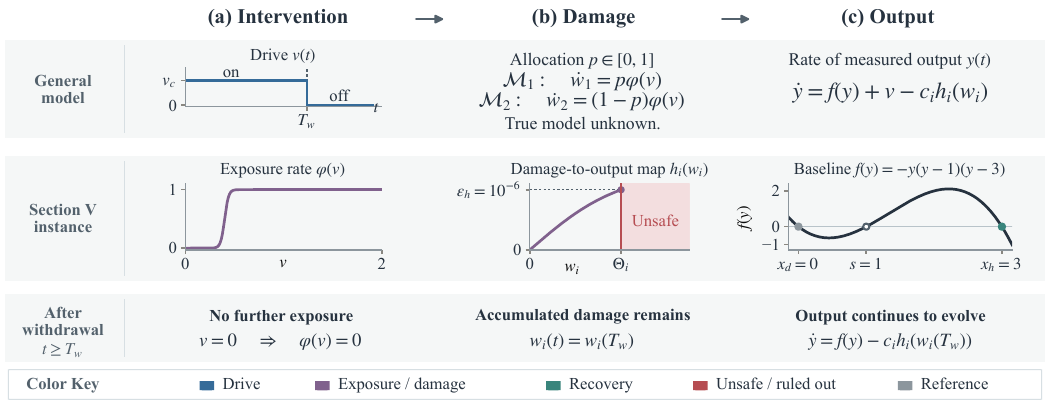}
\caption{Drive promotes recovery but also generates damage. (a) Exposure accumulates at rate $\varphi(v)$. The drive trace is schematic. Withdrawal sets this rate to zero. (b) Allocation $p=0$ prevents damage under $\md1$, while $p=1$ prevents it under $\md2$, so the zero-damage allocation depends on the unknown candidate. Withdrawal stops further accumulation but leaves existing damage. (c) Damage affects measured output through the drag $c_ih_i$. Since $h_i\le\eps_h$ even at tolerance $\Theta_i$, damage close to failure can have little effect on output. Residual drag shifts the separator after withdrawal. Safe withdrawal above $s+a$ leaves the output in the recovery basin by Lemma~\ref{lem:persistence}. The middle row shows the Section~\ref{sec:instance} maps. The color key applies to both figures.}
\label{fig:mechanism}
\end{figure*}

\begin{itemize}
\item A nonlinear exposure certificate shows that every basin crossing requires at least $\Jint$ exposure, even if output pauses or reverses (Lemma~\ref{lem:min-exposure}). If this exceeds total tolerance, no common safe open-loop input restores either candidate (Theorem~\ref{thm:impossibility}).
\item A sufficient restoration certificate uses proportional allocation to restore both candidates without identification when total tolerance exceeds suitable comparison-drive exposure $\bar J$ (Proposition~\ref{prop:capacity}).
\item Theorem~\ref{thm:risk-exchange} bounds sampled-feedback safe-restoration probability under one candidate by alternative-model failure risk plus an explicit correction for noise, sample count and horizon when $\Jint$ exceeds total tolerance.
\end{itemize}

\noindent\emph{Related work.}
SafeMDP expands a certified safe set through safely returnable states \cite{turchetta-safemdp}. Here, returning the output does not undo accumulated damage. Reachability-based safety filters preserve a robust invariant set \cite{fisac-safety}. Our bound identifies when safety under both candidates prevents restoration. Active fault diagnosis designs inputs to separate bounded-error reachable-output sets \cite{scott-afd}. We account for damage from diagnostic and restorative inputs.

Zeng et al.\ \cite{zeng2023hardness} derive sample requirements from statistically close linear systems with no common stabilizing static gain. Their gain acts indefinitely, whereas our intervention ends and recovery must persist. Saligrama \cite{saligrama-commit} bounds the performance lost through model uncertainty using observations available before an action rules out safe continuation under an alternative. Evidence produced by that action arrives too late to guide its choice. This viewpoint includes irreversible damage and finite tolerance. Our noisy analysis uses a similar comparison of observations up to the first candidate failure under a common input. For nonlinear basin transfer, we derive a lower bound on crossing exposure from the drift and damage-rate maps. This bound can rule out common safe restoration, while a proportional allocation certifies lasting recovery after withdrawal when tolerance is sufficient.

\section{Problem formulation}\label{sec:formulation}

The input $u=(v,p)\in U:=[0,\vmax]\times[0,1]$ consists of drive $v$ and allocation $p$. The signals $(y^{(i)},w_i,v^{(i)},p^{(i)})$ describe operation under $\md{i}$. Superscripts distinguish signals that can differ under feedback.

\Needspace{3\baselineskip}
\begin{assumption}[Bistable output dynamics]\label{ass:f}
The undamaged, uncontrolled system obeys $\dot y=f(y)$, where $f:\R\to\R$ is continuously differentiable ($C^1$) and satisfies the following.
\begin{enumerate}
\item[(i)] \emph{Bistability.} $f$ has exactly three zeros $x_d < s < x_h$, all hyperbolic with $f'(x_d)<0$, $f'(s)>0$, $f'(x_h)<0$. The stable equilibria $x_d$ and $x_h$ represent degraded and desired operation, and $s$ separates their basins.
\item[(ii)] \emph{Bounded operation.} $f>0$ on $(-\infty,x_d)$, $f<0$ on $(x_h,\infty)$, and there exist $y^- < x_d$, $y^+ > x_h$ with $f(y) \ge \bar\rho$ for $y \le y^-$ and $f(y) + \vmax \le 0$ for $y \ge y^+$, where $\bar\rho>0$ is fixed in Assumption~\ref{ass:models}. These bounds keep the output in $\mathcal Y := [y^-, y^+]$ until failure.
\item[(iii)] \emph{Crossable barrier.} Full drive $\vmax$ exceeds the barrier level $u_a := \max_{[x_d,s]}(-f) > 0$ at $x^\ast \in (x_d,s)$, enabling undamaged crossing of $s$.
\item[(iv)] \emph{Separation rate.} On $\mathcal Y$, the one-sided Lipschitz constant $\nu_f := \sup_{\mathcal Y} f' \ge f'(s) > 0$ bounds how fast outputs under a common input can separate.
\end{enumerate}
\end{assumption}

\Needspace{3\baselineskip}
\begin{assumption}[Damage-rate map]\label{ass:damage-rate}
Damage is generated by the drive alone and is small at low drive. A continuous map $g$ on a nonnegative interval containing zero is called a class-$\mathcal K$ function if it is strictly increasing and $g(0)=0$. The rate $\varphi:[0,\vmax]\to[0,\phimax]$ is real-analytic and of class $\mathcal K$, with Lipschitz constant $L_\varphi$, so $\varphi>0$ on $(0,\vmax]$. On a low-drive range $[0,\bar v]$, $\bar v\in(0,u_a)$, it satisfies $\varphi \le \eps_\varphi$.
\end{assumption}

\Needspace{3\baselineskip}
\begin{assumption}[Candidate models]\label{ass:models}
The candidates share input and output spaces but differ in damage routing. Candidate $\md{i}$ has state $(y,w_i)\in\R\times[0,\infty)$ and dynamics
\begin{alignat}{2}
\md1:\;&\dot y=f(y)+v-c_1h_1(w_1),\quad &\dot w_1&=p\varphi(v),\label{eq:M1}\\
\md2:\;&\dot y=f(y)+v-c_2h_2(w_2),\quad &\dot w_2&=(1{-}p)\varphi(v),\label{eq:M2}
\end{alignat}
with known $c_i>0$, $y(0)=x_d$, $w_i(0)=0$, output $y$, and failure set $\{w_i\ge\Theta_i\}$ for known $\Theta_i>0$. The damage-to-output maps $h_i$ are real-analytic and of class $\mathcal K$, with the \emph{pre-failure output bound}
\begin{equation}\label{eq:visibility}
h_i(w) \le \eps_h \qquad \text{for all } w \in [0,\Theta_i]
\end{equation}
for a common $\eps_h>0$. Damage can thus approach its tolerance while producing only a small drag, $c_ih_i(w_i)$. The drag bounds are $\rho_i := c_i\eps_h$, $\rho := \rho_1+\rho_2$, and $\bar\rho := \max_i\rho_i$. Section~\ref{sec:indist} converts them into a measured-output separation bound.
\end{assumption}

Damage is irreversible because its rate is nonnegative. All parameters are known, and only the candidate is uncertain. A guarantee over any larger model class containing this pair must hold for both.

\Needspace{3\baselineskip}
\begin{definition}[Inputs and failure]\label{def:policy}
An open-loop input is a measurable function $u:[0,\infty)\to U$. Locally Lipschitz dynamics and bounded input give a unique absolutely continuous solution up to failure. The equations hold almost everywhere, allowing discontinuous input switches. The failure time is $\tau_i:=\inf\{t:w_i(t)\ge\Theta_i\}$, with $\tau_i=\infty$ if the tolerance is never reached. Theorem~\ref{thm:risk-exchange} introduces sampled output feedback.
\end{definition}

Invariance in Assumption~\ref{ass:f}(ii) holds because the field points inward at both ends. At $y^+$, $\dot y\le f(y^+)+\vmax\le0$, while at $y^-$, $\dot y\ge\bar\rho-c_ih_i(w_i)\ge\bar\rho-\rho_i\ge0$ by \eqref{eq:visibility}.

\Needspace{3\baselineskip}
\begin{definition}[Viability and restoration]\label{def:viable}
Safety here means keeping the damage below the tolerance of the applicable candidate. A policy $\pi$ is \emph{viable} under $\md{i}$ on $[0,T]$ if $\tau_i>T$. It is \emph{robustly viable} if this holds under both $\md1$ and $\md2$. Robust viability evaluates one decision rule under both candidates. An open-loop rule produces the same input. Feedback may select different inputs because the candidates generate different observations.

A policy is \emph{restorative} under $\md{i}$ on $[0,T]$ with margin $a$ if some $T_w\le T$ satisfies $y^{(i)}(T_w)\ge s+a$ and $v^{(i)}(t)=0$ for almost every $t>T_w$. After withdrawal, fixed damage leaves a constant drag that can shift the separator upward. To bound this shift, let $m,a_0>0$ be such that the closed $a_0$-neighborhoods of $x_d,s,x_h$ are disjoint and lie in the interior of $\mathcal Y$, with $|f'|\ge m/2$ on each. For $a\in(0,a_0)$, set $\mu(a):=\min\{|f(x)|: x\in\mathcal Y,\ \operatorname{dist}(x,\{x_d,s,x_h\})\ge a\}$ and $\rho^\ast(a):=\min\{\mu(a)/2,\ ma/4\}$. Restoration uses only \emph{admissible} margins, those with $a\in(0,a_0)$ and $\bar\rho\le\rho^\ast(a)$.
\end{definition}

\Needspace{3\baselineskip}
\begin{lemma}[Persistence]\label{lem:persistence}
Under Assumption~\ref{ass:f}, for $a\in(0,a_0)$ and every constant $\kappa \in [0, \rho^\ast(a)]$, the field $f_\kappa := f - \kappa$ has exactly three zeros in $\mathcal Y$, given by $x_d(\kappa) \in [x_d - \frac{a}{2}, x_d]$, $s(\kappa) \in [s, s + \frac{a}{2}]$, $x_h(\kappa) \in [x_h - \frac{a}{2}, x_h]$, each hyperbolic with unchanged stability type. Every solution of $\dot y = f_\kappa(y)$ in $\mathcal Y$ with $y(0) \ge s + a$ converges to $x_h(\kappa)$, and with $y(0) \le s - a$ to $x_d(\kappa)$.
\end{lemma}

\begin{proof}
We bound the equilibrium shifts using monotonicity, then use the drift signs to establish convergence. First, $f_\kappa$ is strictly increasing near $s$. Its values satisfy $f_\kappa(s)=-\kappa\le0$ and $f_\kappa(s+2\kappa/m)\ge(m/2)(2\kappa/m)-\kappa=0$. Hence it has a unique zero $s(\kappa)\in[s,s+2\kappa/m]\subseteq[s,s+a/2]$, using $\kappa\le ma/4$. This zero is unstable because $f_\kappa'=f'\ge m/2$ there. The same argument with decreasing $f$ at $x_d$ and $x_h$ gives unique stable zeros in $[x_d-a/2,x_d]$ and $[x_h-a/2,x_h]$.

Second, these are the only zeros. At distance at least $a$ from the original equilibria, $|f_\kappa|\ge\mu(a)-\kappa>0$. At distances in $(a/2,a)$, monotonicity gives $|f|>ma/4\ge\kappa$, which also excludes a zero. Finally, the unchanged sign pattern makes trajectories move toward the stable equilibrium on their side of $s(\kappa)$. Since $s-a<s(\kappa)<s+a$, both stated initial conditions lie in the required basins.
\end{proof}

After withdrawal, $\varphi(0)=0$ keeps damage constant and fixes the residual drag $\kappa:=c_ih_i(w_i(T_w))\in[0,\bar\rho]$. A safe intervention withdrawn above $s+a$ therefore remains viable forever and converges to an equilibrium in $H=[x_h-a,x_h+a]$. Withdrawal below $s-a$ leads back to the degraded equilibrium.

\section{Exposure and safe restoration}\label{sec:exposure}

For a drive $v$, define cumulative exposure $J_t[v]:=\int_0^t\varphi(v)\dd\sigma$. Its allocated parts are $J_{1,t}:=\int_0^t p\varphi(v)\dd\sigma$ and $J_{2,t}:=\int_0^t(1-p)\varphi(v)\dd\sigma$. Since $J_{i,t}=w_i(t)$,
\begin{equation}\label{eq:exposure-decomp}
J_{1,t}+J_{2,t}=J_t[v] .
\end{equation}
Under a common input, $w_1$ and $w_2$ describe damage in the two candidate models. Only one model describes the plant, but safety under both requires each damage value to stay below its own tolerance. Allocation changes these values without changing total exposure. Basin transfer requires a minimum total exposure regardless of allocation.

\Needspace{3\baselineskip}
\begin{lemma}[Minimum exposure for basin transfer]\label{lem:min-exposure}
For $y\in[x_d,s]$, define the exposure per unit of forward progress by
\begin{equation}\label{eq:Jint}
\begin{aligned}
q(y)&:=\inf_{-f(y)<v\le\vmax}\frac{\varphi(v)}{f(y)+v},\\
\Jint&:=\int_{x_d}^{s}q(y)\dd y.
\end{aligned}
\end{equation}
Under either candidate and any policy, a trajectory from $x_d$ to $y(T_w)\ge s$ without prior failure satisfies $J_{T_w}[v]\ge\Jint$. For $\delta\in(0,u_a)$, let $[a_\delta,b_\delta]$ be the component containing $x^\ast$ of $\{-f\ge u_a-\delta\}$ in $[x_d,s]$, with length $\ell_\delta$. Then
\begin{equation}\label{eq:Jesc}
\Jint\ge\varphi(u_a-\delta)
\frac{\ell_\delta}{\vmax-u_a+\delta}
=: \Jesc(\delta).
\end{equation}
\end{lemma}

\begin{proof}
We integrate a barrier-crossing potential. Extend $q$ by zero outside $[x_d,s]$ and set $Q(y)=\int_{x_d}^y q(r)\dd r$. The admissible choice $v=\vmax$ bounds $q$ between zero and $\varphi(\vmax)/(\vmax-u_a)$, so $Q$ is Lipschitz. In the barrier interior, the ratio diverges at zero denominator, so minimization gives continuous $q$. The chain rule applies along the absolutely continuous output. When $f(y)+v>0$, nonnegative damage gives
\begin{equation}\label{eq:potential-progress}
\frac{\dd}{\dd t}Q(y)=q(y)\dot y
\le q(y)(f(y)+v)\le\varphi(v).
\end{equation}
When $f(y)+v\le0$, the potential cannot increase. It is constant outside the barrier. Integration gives $J_{T_w}\ge Q(y(T_w))-Q(x_d)=\Jint$, including waiting and backtracking. On the stated band, positive progress requires $v>u_a-\delta$. Monotonicity bounds the numerator below by $\varphi(u_a-\delta)$ and the denominator above by $\vmax-u_a+\delta$, proving \eqref{eq:Jesc}.
\end{proof}

The integral uses the full barrier, while the simpler band estimate in Figure~\ref{fig:capacity}(a,b) explains how crossing time and damage rate combine. Bounded drive limits crossing speed. With a saturating damage rate, unbounded drive could make crossing exposure arbitrarily small. The regime is exposure-obstructed when $\Jint>\Theta_1+\Theta_2$.

\Needspace{3\baselineskip}
\begin{theorem}[Impossibility of robust restoration]\label{thm:impossibility}
If $\Jint > \Theta_1 + \Theta_2$, then no robustly viable open-loop input is restorative under \emph{either individual} candidate, for any admissible margin $a$, any horizon, and any $\eps_h>0$ satisfying the standing assumptions.
\end{theorem}

\begin{proof}
The result follows by contradiction. A common input viable under both models satisfies $J_{i,T}=w_i(T)<\Theta_i$. Restoration under either model requires crossing $s$, so Lemma~\ref{lem:min-exposure} gives $J_T[v]\ge\Jint$. But \eqref{eq:exposure-decomp} gives $J_T[v]=w_1(T)+w_2(T)<\Theta_1+\Theta_2<\Jint$, a contradiction.
\end{proof}

\Needspace{3\baselineskip}
\begin{remark}[Non-complementary routing]\label{rem:routing}
The crossing argument also applies to continuous rates $\dot w_i=\phi_i(v,p)\ge0$ with $\phi_i(0,p)=0$ and $\phi_1+\phi_2\ge\phi_{\min}(v)\ge0$ for nondecreasing $\phi_{\min}$.

\Needspace{3\baselineskip}
With the same output assumptions, an input safe under both candidates and restorative under either must satisfy
\begin{equation}\label{eq:routing-exposure}
\Theta_1+\Theta_2>\int_0^{T_w}(\phi_1+\phi_2)\dd t
\ge\Jesc^{\min}(\delta),
\end{equation}
where $\Jesc^{\min}$ is \eqref{eq:Jesc} with $\phi_{\min}$ replacing $\varphi$. Thus $\Jesc^{\min}(\delta)>\Theta_1+\Theta_2$ still implies the obstruction of Theorem~\ref{thm:impossibility}.
\end{remark}

For complementary routing, choose a fixed drive on $[0,\bar T_w]$ and a comparison trajectory starting at $z(0)=x_d$, with
\begin{equation}\label{eq:Jbar}
\dot z=f(z)+v-\bar\rho,\qquad
\bar J:=J_{\bar T_w}[v].
\end{equation}

\Needspace{3\baselineskip}
\begin{proposition}[Sufficient tolerance without identification]\label{prop:capacity}
If $z$ stays in $\mathcal Y$, $z(\bar T_w)\ge s+a$, and $\bar J<\Theta_1+\Theta_2$, then this drive with $p\equiv\Theta_1/(\Theta_1+\Theta_2)$, withdrawn at $\bar T_w$, is viable and restorative under both candidates for every $T\ge\bar T_w$.
\end{proposition}

\begin{proof}
We combine scalar comparison with proportional allocation. Before failure, the drag is at most $\bar\rho$, so each output satisfies $y^{(i)}\ge z$. Proportional allocation gives
\begin{equation}\label{eq:proportional-damage}
w_i(t)=\frac{\Theta_i}{\Theta_1+\Theta_2}J_t[v]
\le\frac{\Theta_i}{\Theta_1+\Theta_2}\bar J<\Theta_i.
\end{equation}
Thus failure cannot occur before withdrawal, when both outputs are at least $s+a$. Lemma~\ref{lem:persistence} then gives restoration.
\end{proof}

Figure~\ref{fig:capacity}(d) compares the Section~\ref{sec:instance} instance with these two certificates.

Knowing the model permits restoration even below $\Jint$. The allocation $p=0$ for $\md1$ or $p=1$ for $\md2$ keeps damage at zero.\label{prop:model-specific} A constant drive $V\in(u_a,\vmax]$, withdrawn at $T_w=(s+a-x_d)/(V-u_a)$, reaches $s+a$ because $\dot y\ge V-u_a$ up to that level. Lemma~\ref{lem:persistence} then gives restoration. The same input fails under the other candidate before withdrawal if $\varphi(V)T_w>\max_i\Theta_i$.

\section{Output separation and noisy feedback}\label{sec:indist}

Observations can reveal which allocation is safe if the candidates' outputs separate enough before failure. We first bound this separation under a common input, then analyze feedback from noisy samples.

\Needspace{3\baselineskip}
\begin{theorem}[Pre-failure output separation]\label{thm:indist}
Under a common measurable input, for $t\le\min(\tau_1,\tau_2)$,
\begin{equation}\label{eq:indist-bound}
|y^{(1)}(t)-y^{(2)}(t)|\le D_0(t)
:=\rho\frac{e^{\nu_f t}-1}{\nu_f}.
\end{equation}
\end{theorem}
\Needspace{3\baselineskip}
\begin{proof}
We bound the growth of the output difference and integrate. Let $e=y^{(1)}-y^{(2)}$. Before failure both outputs lie in $\mathcal Y$. For almost every time $t$ at which $e(t)\ne0$, the mean value theorem gives
\begin{equation}\label{eq:output-differential}
\frac{\dd}{\dd t}|e|=\operatorname{sign}(e)\dot e
\le\nu_f|e|+\rho.
\end{equation}
The common drive cancels and the damage terms are bounded by $\rho_1+\rho_2$. At almost every zero of $e$, the derivative of $|e|$ is zero, so the inequality remains valid. Integration with $e(0)=0$ proves \eqref{eq:indist-bound}. Negative slopes contract the difference, so only the upper bound on $f'$ enters.
\end{proof}

\begin{figure}[t]
\centering
\includegraphics[width=\columnwidth]{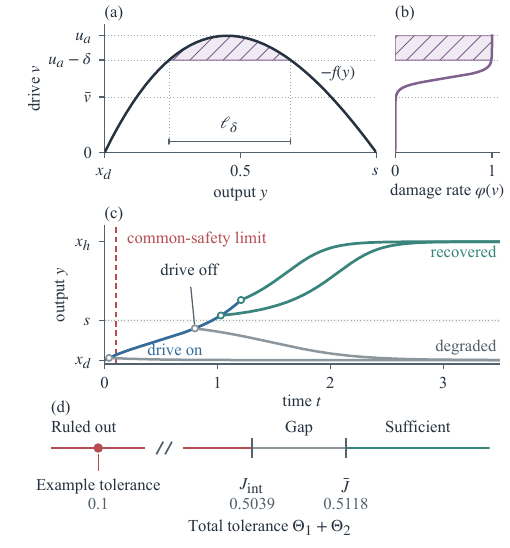}
\caption{Crossing exposure and withdrawal outcomes. (a,b) Crossing the band of width $\ell_\delta$ in Lemma~\ref{lem:min-exposure} requires drive above $u_a-\delta$, whose matching damage-rate range is shown. The shaded areas are not exposure integrals. (c) Under $\md1$ with $v=\cFanDrive$ and $p=0$, recovery depends on withdrawing above $s$. At the red line $t=\cFanCapacityTime$, total tolerance is exhausted under every allocation. Candidate $\md2$ fails earlier at $t=\cFanAltFailureTime$. (d) The Section~\ref{sec:instance} tolerance $\cThetaSum$ lies below $\Jint$, ruling out common safe restoration. Above $\bar J$, proportional allocation certifies it. Neither certificate resolves the gap, which the axis break enlarges.}
\label{fig:capacity}
\end{figure}

\noindent\emph{Observation protocol.}
The controller observes $m_k=y(t_k)+\sigma\xi_k$ at $N$ sampling times $t_k=k\,\delta t\le T$, where $\sigma>0$ and the $\xi_k$ are independent standard Gaussian variables. The controller may randomize through a seed $U$ whose law is fixed and independent of the model and noise. On $[t_k,t_{k+1})$, the input is a jointly measurable function of $(t,m_{1:k},U)$, with well-posed loops under each candidate. Here, $t_0=0$ and $t_{N+1}=T$, so before the first sample the input depends only on $(t,U)$. The controller receives no separate failure indicator, and failure is allowed.

\Needspace{6\baselineskip}
\begin{theorem}[Noise-robust risk exchange]\label{thm:risk-exchange}
Under this protocol and candidate $\md{i}$, let $\alpha_i$ be the probability of failure by $T$ and $1-\beta_i$ the probability of viable restoration with margin $a$ on $[0,T]$. If $\Jint>\Theta_1+\Theta_2$, then
\begin{equation}\label{eq:risk-exchange}
\beta_i+\alpha_{3-i}\ge1-\frac{D_0(T)\sqrt N}{2\sigma},\qquad i\in\{1,2\}.
\end{equation}
\end{theorem}

\Needspace{3\baselineskip}
\begin{proof}
We compare the candidates along a shared observation record, then transfer a failure event. Let $P_i$ denote the law under candidate $i$, including the random seed. First, let $\mathcal F_k$ be generated by $U$ and the first $k$ samples. The input before $t_k$ depends only on $\mathcal F_{k-1}$. Since damage depends on the input alone through $w_i(t)=J_{i,t}$, this past determines both damage paths up to $t_k$. In particular,
\begin{equation}\label{eq:retention-event}
E_k:=\{\min(\tau_1,\tau_2,T)\ge t_k\}\in\mathcal F_{k-1}.
\end{equation}
Under $P_i$, $\tau_i$ is the plant's failure time. The other time, $\tau_{3-i}$, is when the same input would cause failure if $\md{3-i}$ were the plant.

Second, form the stopped record $(U,\tilde m_{1:N})$, retaining $m_k$ on $E_k$ and using an absorbing marker $\star\notin\R$ otherwise. Truncation serves only to compare probabilities and leaves the policy unchanged. Since the stopped past determines retention, conditional on that past $\tilde m_k$ has law $\mathcal N(g_i^k,\sigma^2)$ on $E_k$ and $\delta_\star$ otherwise. Here, $g_i^k$ is candidate $i$'s output at $t_k$ for the input selected by that past. On $E_k$, both hypothetical trajectories follow this same input, so Theorem~\ref{thm:indist} gives $|g_1^k-g_2^k|\le D_0(T)$, including at failure.

\Needspace{3\baselineskip}
The common seed contributes no KL divergence. The chain rule and conditional Gaussian divergences give
\begin{align}
\mathrm{KL}(P_1^{\mathrm{stop}}\|P_2^{\mathrm{stop}})
&=\sum_{k=1}^N\mathbb E_1\!\left[\mathbf1_{E_k}
\frac{(g_1^k-g_2^k)^2}{2\sigma^2}\right]\label{eq:stopped-kl-chain}\\
&\le\frac{ND_0(T)^2}{2\sigma^2}.\label{eq:stopped-kl-bound}
\end{align}
Pinsker's inequality then bounds the total variation distance of the stopped laws by $D_0(T)\sqrt N/(2\sigma)$.

\Needspace{3\baselineskip}
Finally, let $B:=\{\tau_2\le\min(\tau_1,T)\}$. The stopped record determines this event. Each retained past fixes both damage paths until the next sample or $T$, including which tolerance is reached first and any tie. Let $A_1$ be safe restoration under $\md1$ by $T$. On $A_1$, $\tau_1>T$ and Lemma~\ref{lem:min-exposure} gives $J_{T_w}[v]\ge\Jint$. Since $w_1(T_w)<\Theta_1$, the common input gives $w_2(T_w)>\Jint-\Theta_1>\Theta_2$. Thus $A_1\subseteq B$. Under $P_2$, $B$ entails physical failure, so
\begin{align}
1-\beta_1 &\le P_1(B)\le P_2(B)+\mathrm{TV},\label{eq:restoration-event}\\
P_2(B)+\mathrm{TV} &\le\alpha_2+\frac{D_0(T)\sqrt N}{2\sigma}.\label{eq:risk-transfer}
\end{align}
Swapping the model indices proves \eqref{eq:risk-exchange} for $i=2$.
\end{proof}

More samples, less noise, or a longer horizon weaken this necessary risk restriction and can make it vacuous. Retaining the KL divergence before applying Pinsker's inequality gives a tighter bound.

\Needspace{3\baselineskip}
\begin{corollary}[Tighter risk boundary]\label{cor:binary-risk}
Under the hypotheses of Theorem~\ref{thm:risk-exchange}, set $\mathcal K_N:=ND_0(T)^2/(2\sigma^2)$ and define the binary KL divergence
\begin{equation}\label{eq:binary-entropy}
d(r\|a):=r\log\frac{r}{a}+(1-r)\log\frac{1-r}{1-a},
\end{equation}
with natural logarithms and the usual limiting values at the endpoints. If $1-\beta_1>\alpha_2$, then
\begin{equation}\label{eq:binary-risk}
d(1-\beta_1\|\alpha_2)\le\mathcal K_N.
\end{equation}
The bound also holds with the model indices exchanged.
\end{corollary}

\begin{proof}
We reduce the stopped record to whether $B$ occurs, then use monotonicity. The data-processing inequality \cite{kaufmann2016complexity} says that this reduction cannot increase KL divergence, giving $d(r_B\|a_B)\le\mathcal K_N$, where $r_B=P_1(B)$ and $a_B=P_2(B)$. Safe restoration under candidate 1 requires $B$, while under candidate 2 this event implies failure. Hence $r_B\ge1-\beta_1>\alpha_2\ge a_B$. For $r>a$, the binary KL divergence increases with $r$ and decreases with $a$. These monotonicities give \eqref{eq:binary-risk}, including its limiting cases.
\end{proof}

For a required safe-restoration probability $r$, scalar bisection in $d(r\|a)=\mathcal K_N$ finds the minimum failure allowance $a<r$ under the alternative. Attaining this boundary requires a separate policy construction.

\Needspace{3\baselineskip}
\begin{remark}[Almost-sure safety under Gaussian observations]\label{rem:hard-safety}
Under the hypotheses of Theorem~\ref{thm:risk-exchange}, zero failure probability under both candidates rules out restoration under either, however small the positive noise level. More precisely, $\alpha_2=0$ implies $\beta_1=1$, and $\alpha_1=0$ implies $\beta_2=1$. Corollary~\ref{cor:binary-risk} gives this directly, since $d(r\|0)=\infty$ for $r>0$, whereas $\mathcal K_N$ is finite. Better sensing can relax the restriction on positive risk allowances, but cannot remove this zero-risk obstruction with finitely many Gaussian samples.
\end{remark}

\Needspace{3\baselineskip}
\begin{remark}[Exact sensing]\label{rem:gain}\label{prop:precursor}
Exact continuous sensing can identify the safe allocation at arbitrarily small exposure. A probe with $v=V_0\in(0,\vmax]$, $p=1$, and duration $t_p>0$ is safe if $\varphi(V_0)t_p<\Theta_1$. Its exact residual
\begin{equation}\label{eq:exact-residual}
R(t_p)=y(t_p)-y(0)-\int_0^{t_p}[f(y)+V_0]\dd t
\end{equation}
is negative under $\md1$ since $h_1$ is of class $\mathcal K$, and zero under $\md2$, where the probe causes no damage. The identified zero-damage allocation then stops further accumulation. Comparison with $f-\bar\rho$ keeps the probe above $x_d-a/2$.

\Needspace{3\baselineskip}
If $\vmax>u_a+\bar\rho$, an admissible constant drive above $u_a+\bar\rho$ then reaches $s+a$ in finite time, allowing restorative withdrawal. This test requires a known field and exact output integration which noisy samples do not provide.
\end{remark}

\Needspace{3\baselineskip}
\section{Certified bistable instance}\label{sec:instance}

We evaluate the certificates for $f(y)=-y(y-1)(y-3)$, with $x_d=0$, $s=1$, $x_h=3$, and $\mathcal Y=[-1,4]$. The derivative $f'$ is a concave parabola peaking at $y=4/3$, so $\nu_f=7/3$ globally, whereas $\sup_{\mathcal Y}|f'|=19$. The damage-rate map is the saturating analytic ramp $\varphi(v) = \phimax\,[\tanh(k(v-v_0)) + \tanh(kv_0)]/[\tanh(k(\vmax-v_0)) + \tanh(kv_0) ]$. The maps $h_i(w)=\eps_h\tanh(w/\Theta_i)/\tanh 1$, with $c_i=1$, are real-analytic, of class $\mathcal K$, and satisfy \eqref{eq:visibility}. Table~\ref{tab:constants} gives the design settings and resulting certificates.

The steep rise in damage rate separates low-drive operation from basin transfer. Low drive generates little exposure, while overcoming the barrier requires substantial exposure. The persistence condition also holds. Choosing $a_0=\cPersistRadius$ gives $\rho^\ast(a)\approx\cPersistBound>\bar\rho$.

\begin{table}[t]
\centering\footnotesize
\setlength{\tabcolsep}{2pt}
\caption{Parameters and exposure certificates of the instance}
\label{tab:constants}
\begin{tabular}{@{}lll@{}}
\toprule
Symbol & Meaning & Value \\
\midrule
\multicolumn{3}{@{}l}{\emph{Model and design settings}} \\
$\vmax,\ V$ & maximum, known-model drive & $2,\ 1.5$ \\
$\phimax,\ v_0,\ k$ & damage-rate parameters & $1,\ 0.4,\ 25$ \\
$\Theta_1,\ \Theta_2$ & damage tolerances & $0.05,\ 0.05$ \\
$a,\ u_a-\delta$ & withdrawal margin, band level & $0.05,\ 0.5$ \\
\midrule
\multicolumn{3}{@{}l}{\emph{Computed quantities}} \\
$u_a$ & barrier level & $0.631$ \\
$a_\delta,\ b_\delta$ & crossing-band endpoints & $0.237,\ 0.685$ \\
$\Jint$ & crossing lower bound \eqref{eq:Jint} & $0.5039$ \\
$\bar J$ & sufficient exposure \eqref{eq:Jbar} & $0.5118$ \\
$T_w,\ \bar T_w$ & known-model, common withdrawal & $1.21,\ 4.726$ \\
\bottomrule
\end{tabular}

\smallskip
\parbox{\columnwidth}{\raggedright\footnotesize Computed values use three or four significant figures, with $\eps_h=\cEpsH$.}
\end{table}

The instance has total tolerance $\Theta_1+\Theta_2=\cThetaSum$, below $\Jint\approx\cJint$, so Theorem~\ref{thm:impossibility} rules out common safe restoration. A comparison drive with levels $(\cCertDrives)$ and durations $(\cCertDurations)$ reaches $z(\bar T_w)>s+a$ at $\bar T_w=\cCertTw$. Its exposure is $\bar J\approx\cJconstruct$, only $\cCertGapPercent\%$ above the lower bound. With the same drag bound $\bar\rho$, Proposition~\ref{prop:capacity} certifies this drive and proportional allocation whenever total tolerance exceeds $\bar J$, for every $T\ge\bar T_w$, with withdrawal at $\bar T_w$.

This instance also illustrates sensing quality. With $T=\cHorizon$ and $N=\cNsamp$, the output-separation bound is $D_0(T)\approx\cDzero$. At noise standard deviation $\sigma=0.01$, the lower bound on $\beta_1+\alpha_2$ is approximately $0.955$. Reducing $\sigma$ tenfold to $0.001$ lowers the bound to $0.549$. If safe restoration under candidate 1 is required with probability $0.9$, the binary KL boundary requires allowing failure under candidate 2 with probability at least approximately $0.871$ and $0.476$, respectively.
 These sensing calculations use $T=\cHorizon$, not the longer constructive horizon. Exposure tightness does not establish tightness of the noisy risk bound. At either positive noise level, Remark~\ref{rem:hard-safety} still rules out restoration under almost-sure safety. The code linked on the first page reproduces the numerical results in this section.

\Needspace{3\baselineskip}
\section{Discussion}\label{sec:discussion}

Treatment withdrawal in CML raises the question of how to sustain disease control while limiting treatment harm. In our construction, accumulated damage shifts the separator and therefore affects recovery after withdrawal. The exposure certificates distinguish a structural obstacle to common safe restoration from failure of a particular schedule. The undecided range between $\Jint$ and $\bar J$ reflects two different calculations. The lower one omits drag and stops at $s$, whereas the construction reaches $s+a$ under the drag bound $\bar\rho$ with one fixed schedule.

\Needspace{3\baselineskip}
Learning before withdrawal can inform the choice of intervention. H\"ahnel et al.\ find that simulated responses to dose reduction help predict withdrawal outcomes \cite{haehnel-cml}. Here, the unknown candidate determines the safe allocation. With noisy observations, learning depends on how much the candidates' outputs separate before failure. In the exposure-obstructed regime, Theorem~\ref{thm:risk-exchange} bounds safe-restoration probability through failure risk under the alternative and the quality of finite noisy observations. One design to study uses the safe probe of Remark~\ref{rem:gain} with $p=1$, then collects samples with $v=0$. The candidates retain different drag without accumulating further damage. How probe exposure, noise and observation time affect safe-restoration probability under fixed horizon and failure allowance remains open.

\Needspace{3\baselineskip}
The recovery dynamics in H\"ahnel et al.\ couple leukemic and immune populations \cite{haehnel-cml}, whereas our candidates share a scalar baseline drift. With several recovery variables, paths into the recovery basin could differ in exposure cost. A lower certificate would have to account for the least costly path. Allowing uncertainty in those dynamics would also enlarge the candidate class. Under the same control specification, a common input cannot guarantee safe restoration for every member of a class containing the obstructed pair, since it would have to do so for that pair. Which routing differences matter for choosing a safe allocation is a further question.

Preparing a system for withdrawal and learning how to intervene safely are thus parts of the same control problem. The exposure and risk certificates provide a way to assess whether the pursuit of lasting recovery is compatible with the safety required along the way.

\bibliographystyle{IEEEtran}
\bibliography{IEEEabrv,refs}

\end{document}